\documentclass[journal]{IEEEtran} 

\usepackage{mathtools} 
\usepackage{booktabs}
\usepackage{tikz}

\usetikzlibrary{arrows,positioning,automata,shapes,calc,patterns}
\usetikzlibrary {arrows.meta,bending,positioning,patterns}

\tikzset{meter/.append style={draw, inner sep=10, rectangle, font=\vphantom{A}, minimum width=30, line width=.8,
		path picture={\draw[black] ([shift={(.1,.3)}]path picture bounding box.south west) to[bend left=50] ([shift={(-.1,.3)}]path picture bounding box.south east);\draw[black,-{Stealth[length=1mm, width=1mm]}] ([shift={(0,.15)}]path picture bounding box.south) -- ([shift={(.2,-.15)}]path picture bounding box.north);}}}

\usepackage{physics}
\usepackage{amssymb}
\usepackage{amsthm}
\usepackage[compress]{cite}

\newtheorem{proposition}{Proposition}
\newtheorem{definition}{Definition}
\newtheorem{lemma}{Lemma}
\newtheorem{theorem}{Theorem}
\newtheorem{corollary}{Corollary}
\newtheorem{example}{Example}

\newcommand{\NOTll}{\hskip 0.4mm \not \hskip -0.4mm \ll}

\DeclareMathOperator*{\supp}{supp}

\newcommand{\X}{\mathbb{X}}
\newcommand{\Y}{\mathbb{Y}}

\allowdisplaybreaks

\title{Barycentric and Pairwise R\'{e}nyi Quantum Leakage}

\author{Farhad Farokhi\thanks{
F. Farokhi is with the Department of Electrical and Electronic Engineering, the University of Melbourne, Parkville, VIC 3010, Australia. e-mail: farhad.farokhi@unimelb.edu.au}}
  
\begin{document}
\maketitle

\begin{abstract}
Barycentric and pairwise quantum R\'{e}nyi leakages are proposed as two measures of information leakage for privacy and security analysis in quantum computing and communication systems. These quantities both require minimal assumptions on the eavesdropper, i.e., they do not make any assumptions on the eavesdropper's attack strategy or the statistical prior on the secret or private classical data encoded in the quantum system. They also  satisfy important properties of positivity, independence, post-processing inequality, and unitary invariance. The barycentric quantum R\'{e}nyi leakage can be computed by solving a semi-definite program and the pairwise quantum R\'{e}nyi leakage possesses an explicit formula. The  barycentric and pairwise quantum R\'{e}nyi leakages form upper bounds on the maximal quantum leakage, the sandwiched quantum $\alpha$-mutual information,  the accessible information, and the Holevo's information. Furthermore, differentially-private quantum channels are shown to bound these measures of information leakage. Global and local depolarizing channels, that are common models of noise in quantum computing and communication, restrict private or secure information leakage. Finally, a privacy-utility trade-off formula in quantum machine learning using variational circuits is developed. The privacy guarantees can only be strengthened, i.e., information leakage can only be reduced, if the performance degradation grows larger and \textit{vice versa}. \end{abstract}

\section{Introduction}\label{sec:intro}
Quantum computing provides various improvements over classical counterparts, such as speed up~\cite{shor1994algorithms}, security~\cite{shor2000simple}, and robustness~\cite{west2023towards}. These advantages have motivated considerable attention towards theory and practice of quantum computing systems. A particularly fruitful direction is quantum machine learning~\cite{biamonte2017quantum}. However, machine learning and data analysis can result in unintended or undesired information leakage~\cite{kearns2019ethical}. As quantum computing hardware becomes more commercially available and quantum algorithm move to the public domain, these privacy and security threats can prove to be detrimental in adoption of quantum technologies particularly for real-world sensitive, private, or proprietary datasets. Therefore, we need to develop rigorous frameworks for understanding information leakage in quantum systems and constructing  secure and private algorithms by minimizing unintended information leakage. Note that the applications of these measures of information leakage is not entirely restricted to quantum machine learning or data privacy. Even in quantum communication, there is a need to understand how much information an eavesdropper can extract from the underlying quantum systems~\cite{farokhiPRA}. 

A common drawback of current notions of information leakage, such as quantum mutual information, accessible information, and Holevo's information, is that they implicitly assume they know the intention of the eavesdropper or the adversary. They assume that the eavesdropper is interested in extracting the entirety of the classical data that is encoded in the state of the quantum system (for communication or analysis). While perfectly reasonable in computing capacity of quantum channels and developing information storage or compression strategies, this assumption can be problematic in security or privacy analysis. In practice, we may not know the intention of the eavesdropper. Imposing extra assumptions on the eavesdropper is akin to underestimating its capabilities, which can be a lethal flaw in security or privacy analysis. Furthermore, in the classical setting, it is shown that mutual information and its derivatives are not suitable measures of information leakage for security and privacy analysis~\cite{issa2019operational}. These observations motivated developing a \textit{maximal} notion of information leakage that is more suited to the task at hand~\cite{issa2019operational,farokhi2021measuring}.

Earlier attempts in developing the corresponding notion of maximal information leakage in quantum systems resulted in maximal quantum leakage~\cite{farokhiPRA}, which was shown to satisfy important properties of positivity (i.e., information leakage is always greater than or equal to zero), independence property (i.e., information leakage is zero if the quantum state is independent of the classical data), and post-processing inequality (i.e., information leakage can only be reduced if an arbitrary quantum channel is applied to the quantum state and therefore additional processing cannot increase information leakage). These properties are cornerstones of axiomatic frameworks for measuring information leakage in classical security analysis~\cite{issa2019operational,farokhi2021measuring}. However, maximal quantum leakage proposed in~\cite{farokhiPRA} did not possess an explicit formula and an iterative algorithm was required to compute it in general, c.f., accessible information~\cite{vrehavcek2005iterative}.

In this paper, we propose two new measures of information leakage, namely, barycentric and pairwise quantum R\'{e}nyi leakages, based on the quantum
R\'{e}nyi divergence or the sandwiched quantum R\'{e}nyi relative entropy~\cite{muller2013quantum} of order $\infty$. These two quantities form upper bounds for the maximal quantum leakage~\cite{farokhiPRA}. They both require minimal assumptions on the eavesdropper, i.e., they do not assume the eavesdropper's attack strategy is known and they do not require priors on the secret or private classical data encoded in the states of the quantum system for communication or analysis. Both barycentric and pairwise quantum R\'{e}nyi leakages satisfy important properties of positivity, independence, and post-processing inequality. They also satisfy unitary invariance, i.e., application of a unitary on the quantum state does not change the information leakage. Unitary invariance in postulated to be important for a quantum measure of information~\cite{muller2013quantum,10106314906367}. These measures are computationally superior to quantum maximal leakage. The barycentric quantum R\'{e}nyi leakage can be computed by solving a semi-definite program while the pairwise quantum R\'{e}nyi leakage possesses an explicit formula. The barycentric and pairwise quantum R\'{e}nyi leakages form upper bounds for the sandwiched quantum $\alpha$-mutual information, the accessible information, and the Holevo's information. Finally, we show that differentially-private quantum channels~\cite{hirche2023quantum} bound the barycentric and pairwise quantum R\'{e}nyi leakages. Therefore, global and local depolarizing channels, that are common models of noise in quantum computing devices and quantum communication systems, are effective quantum channels for bounding private or secure information leakage. Using this result, we develop a privacy-utility trade-off in quantum machine learning using variational circuits. This fundamental trade-off demonstrates that the privacy guarantees can only be strengthened, i.e., information leakage is reduced, if the performance degradation becomes larger and \textit{vice versa}. This is a novel characterization of the trade-off between privacy and utility in quantum machine learning. 

The rest of the paper is organized as follows. We first present some preliminary material on classical and quantum information theory in Section~\ref{sec:prem}. We then formalize the barycentric and pairwise quantum R\'{e}nyi leakages in Section~\ref{sec:info_leakage}. We establish the relationship between these measures of information leakage and quantum differential privacy in Section~\ref{sec:DP}. We investigate privacy-utility trade-off for quantum machine learning in Section~\ref{sec:qml}. We finally conclude the paper in Section~\ref{sec:conc}.

\section{Preliminary Information} \label{sec:prem}
In this section, we review some basic concepts from classical and quantum information theory. A reader with this knowledge may benefit from directly jumping to Section~\ref{sec:info_leakage}.

\subsection{Random Variables and Classical Information}
Random variables are denoted by capital Roman letters, such as $X\in\X$ and $Y\in\Y$. A random variable $X$ is discrete if the set of all its possible outcomes $\X$ is finite. Any discrete random variable $X$ is fully described by its probability mass function  $p_X(x):=\mathbb{P}\{X=x\}>0$ for all $x\in\mathbb{X}$. The support set of random variable $X$ is defined as $\supp(X)=\supp(p_X):=\{x\in\X\,|\,p_X(x)>0\}\subseteq\X$. The set of all probability mass functions with domain $\X$ is $\Delta(\X):=\{\pi:\X\rightarrow\mathbb{R}_{\geq 0}\,|\,\sum_{x\in\X}\pi(x)=1\}$.

For all $\alpha\in(0,1)\cup(1,\infty)$ and probability mass functions $p,q$ such that $\supp(p)\subseteq\supp(q)$, the R\'{e}nyi divergence (or relative entropy) of order $\alpha$~\cite{renyi1961measures,van2014renyi} is 
\begin{align*}
d_\alpha(p\|q):=\frac{1}{\alpha-1}\log\left(\sum_{x\in\supp(q)}p^\alpha(x)q^{1-\alpha}(x) \right).
\end{align*}
All logarithms, in this paper, are in base 2 and, therefore, all the information quantities are measured in bits. On few occasions, we would require logarithm in the natural basis, which is denoted by $\ln(\cdot)$ instead of $\log(\cdot)$. 
By convention, $d_\alpha(p\|q)=\infty$ if $\supp(p)\nsubseteq\supp(q)$ and $\alpha\geq  1$. For $\alpha< 1$, the same definition holds even if $\supp(p)\nsubseteq\supp(q)$. For $\alpha=1,\infty$ (and also $\alpha=0$ which is not used in this paper), we define the R\'{e}nyi divergence  by extension:
\begin{align*}
    d_1(p\|q)&:=\lim_{\alpha\rightarrow 1}d_\alpha (p\|q)=\!\!\!\!\sum_{x\in\supp(q)}p(x)\log\left(\frac{p(x)}{q(x)} \right),\\
    d_\infty(p\|q)&:=\lim_{\alpha\rightarrow \infty}d_\alpha (p\|q)=\log\left(\max_{x\in\supp(q)}\frac{p(x)}{q(x)} \right).
\end{align*}
Note that $d_1(p\|q)$ is the usual Kullback–Leibler (KL) divergence~\cite[\S~2.3]{cover2012elements}. 
Sibson's $\alpha$-mutual information, an extension of mutual information in information theory~\cite{verdu2015alpha}, between random variables $X\in\X$ and $Y\in\Y$ is
\begin{align} \label{eqn:Sibson_info}
    I_\alpha(X;Y):=\inf_{\tilde{q}} d_\alpha(P_{XY}\|P_X\times \tilde{q}),
\end{align}
where $P_{XY}\in\Delta(\X\times\Y)$ is the joint probability mass function for joint random variable $(X,Y)$, $P_X\in\Delta(\X)$ and $P_Y\in\Delta(\Y)$ are the marginal probability mass functions for random variable $X$ and $Y$ separately, and $\tilde{q}\in\Delta(\Y)$ is any general probability mass function. By continuity~\cite{verdu2015alpha}, we get
\begin{align*}
    I_1(X;Y)\!:=&\lim_{\alpha\rightarrow 1}I_\alpha (X;Y)\\
    =&\hspace{-.4in} \sum_{(x,y)\in\supp(p_X)\times\supp(p_Y)} \hspace{-.4in}\!\!\!\!P_{X,Y}(x,y)\log\left(\frac{P_{X,Y}(x,y)}{P_{X}(x)P_{Y}(y)} \right)\!,\\
    I_\infty(X;Y)\!:=&\lim_{\alpha\rightarrow \infty}I_\alpha (X;Y)\\
    =&\log\left(\sum_{y\in\Y}\max_{x\in\supp(X)}P_{Y|X}(y|x) \right)\!,
\end{align*}
where $I_1(X;Y)$ is the common mutual information~\cite[\S~2.3]{cover2012elements}. For a more thorough treatment of the R\'{e}nyi divergence and $\alpha$-mutual information, see~\cite{verdu2015alpha}.

\subsection{Quantum States and Information}
We denote finite-dimensional Hilbert space by $\mathcal{H}$. The set of linear operators from $\mathcal{H}$ to $\mathcal{H}$ is denoted by $\mathcal{L}(\mathcal{H})$. Further, $\mathcal{P}(\mathcal{H})\subset\mathcal{L}(\mathcal{H})$ denotes the set of positive semi-definite operators on Hilbert space $\mathcal{H}$ and $\mathcal{S}(\mathcal{H})\subset\mathcal{P}(\mathcal{H})$ denotes the set of density operators on $\mathcal{H}$, i.e., the set of positive semi-definite operators with unit trace. The state of a quantum system is modelled by a density operator in $\mathcal{S}(\mathcal{H})$. Lower case Greek letters, such as $\rho$ and $\sigma$, are often used to denote density operators or quantum states. A formalism to model quantum measurements is the positive operator-valued measure (POVM), i.e., a set of positive semi-definite matrices $F=\{F_i\}_{i}$ such that $\sum_{i} F_i=I$. POVMs are particularly useful when the post-measurement state of the quantum system is of no interest, i.e., the quantum system is discarded after measurement. For POVM $F=\{F_i\}_{i}$, the probability of obtaining output $i$ when taking a measurement on a system with quantum state $\rho$ is $\trace(\rho F_i)=\trace( F_i \rho)$. This is typically called the Born's rule. A quantum channel is a mapping from the space of density operators to potentially another space of density operators that is both completely positive and trace preserving. Calligraphic capital Roman letters, such as $\mathcal{E}$ and $\mathcal{N}$, are used to denote quantum channels. For a more detailed treatment of basic definitions and properties in quantum information theory, see~\cite{wilde2013quantum}.

For all $\alpha\in(0,1)\cup(1,\infty)$, and arbitrary $\rho\in\mathcal{S}(\mathcal{H})$ and $\sigma\in\mathcal{P}(\mathcal{H})$, the quantum R\'{e}nyi relative entropy~\cite{PETZ198657} is 
\begin{align*}
    D_\alpha (\rho\|\sigma):=\frac{1}{\alpha-1}\log\left(\trace\left(\rho^\alpha\sigma^{1-\alpha}\right)\right),
\end{align*}
if the support set of $\rho$ is contained within the support set of $\sigma$, denoted by $\rho\ll \sigma$, i.e., the kernel of operator $\sigma$ lines within the kernel of operator $\rho$. By convention, $D_\alpha (\rho\|\sigma)=\infty$ if $\rho \NOTll \sigma$ and $\alpha\geq 1$. For $\alpha< 1$, the same definition holds even if $\rho \NOTll \sigma$. For $\alpha=1,\infty$, we can define the quantum R\'{e}nyi relative entropy by extension:
\begin{align*}
    D_1(\rho\|\sigma)&:=\lim_{\alpha\rightarrow 1}D_\alpha (\rho\|\sigma)=\trace\left(\rho(\log(\rho)-\log(\sigma))\right),\\
    D_\infty(\rho\|\sigma)&:=\lim_{\alpha\rightarrow \infty}D_\alpha (\rho\|\sigma)=\log\left(\max_{i,j:\bra{i}\ket{\bar{j}}\neq 0}\frac{v_i}{\lambda_j} \right),
\end{align*}
where $\rho=\sum_{i}v_i \ket{i}\bra{i}$ and $\sigma=\sum_j\lambda_j\ket{\bar{j}}\bra{\bar{j}}$. Note that ${D}_1(\rho\|\sigma)$ is the usual quantum relative entropy~\cite{umegaki1962conditional}.  Similarly, for all $\alpha\in(0,1)\cup(1,\infty)$, and arbitrary $\rho\in\mathcal{S}(\mathcal{H})$ and $\sigma\in\mathcal{P}(\mathcal{H})$ such that $\rho\ll \sigma$, the quantum R\'{e}nyi divergence or the sandwiched quantum R\'{e}nyi relative entropy~\cite{muller2013quantum} is 
\begin{align}
    \widetilde{D}_\alpha (\rho\|\sigma):=&\frac{1}{\alpha-1}\log\left(\trace\left(\left(\sigma^{\frac{1-\alpha}{2\alpha}}\rho\sigma^{\frac{1-\alpha}{2\alpha}}\right)^\alpha\right)\right)\nonumber\\
    =&\frac{1}{\alpha-1}\log\left(\trace\left(\left(\sigma^{\frac{1-\alpha}{\alpha}}\rho\right)^\alpha\right)\right)\nonumber\\
    =&\frac{1}{\alpha-1}\log\left(\trace\left(\left(\rho\sigma^{\frac{1-\alpha}{\alpha}}\right)^\alpha\right)\right).\label{eqn:tildeD:def}
\end{align}
Again, by convention, $\widetilde{D}_\alpha (\rho\|\sigma)=\infty$ if $\rho \NOTll \sigma$ and $\alpha\geq 1$. For $\alpha< 1$, the same definition holds even if $\rho \NOTll \sigma$.  Note that these alternative formulations stem from the fact that, for any two operators $\rho,\sigma\in\mathcal{L}(\mathcal{H})$, $\rho\sigma$ and $\sigma\rho$ have the same eigenvalues~\cite[Exercise~I.3.7]{bhatia2013matrix}. For $\alpha=1,\infty$, we can define the sandwiched quantum R\'{e}nyi relative entropy by extension:
\begin{align*}
    \widetilde{D}_1(\rho\|\sigma)&\!:=\!\!\lim_{\alpha\rightarrow 1}\!\widetilde{D}_\alpha (\rho\|\sigma)\!=\!\trace\left(\rho(\log(\rho)-\log(\sigma))\right),\\
    \widetilde{D}_\infty(\rho\|\sigma)&\!:= \!\!\lim_{\alpha\rightarrow \infty}\!\widetilde{D}_\alpha (\rho\|\sigma)\!=\!\log\left(\inf\{\mu\in\mathbb{R}\!:\! \rho\leq \mu \sigma \} \right),
\end{align*}
where $\widetilde{D}_1(\rho\|\sigma)={D}_1(\rho\|\sigma)$ is the usual quantum relative entropy~\cite{umegaki1962conditional}. 

\begin{lemma} \label{lemma:properties:sandwiched}
The following results hold for the sandwiched quantum R\'{e}nyi relative entropy:
\begin{itemize}
    \item[(a)] \textbf{Post-Processing Inequality}: For any $\rho\in\mathcal{S}(\mathcal{H})$ and $\sigma\in\mathcal{P}(\mathcal{H})$, such that $\rho\ll\sigma$, and any quantum channel $\mathcal{E}$, 
    $\widetilde{D}_\alpha (\mathcal{E}(\rho)\|\mathcal{E}(\sigma))\leq \widetilde{D}_\alpha (\rho\|\sigma)$;
    \item[(b)] \textbf{Order Axiom}: $\widetilde{D}_\alpha (\rho\|\sigma)\leq 0$ if $\rho\preceq \sigma$ and $\widetilde{D}_\alpha (\rho\|\sigma)\geq 0$ if $\rho\succeq \sigma$;
    \item[(c)] \textbf{Unitary Invariance}: $\widetilde{D}_\alpha (U\rho U^\dag\|U\sigma U^\dag)=\widetilde{D}_\alpha (\rho\|\sigma)$.
\end{itemize}
\end{lemma}

\begin{proof}
    The data-processing inequality follows by setting $\alpha\geq 1$ and $z=\alpha$~\cite[Theorem~1]{10106314906367}. The order axiom and unitary invariance follow from~\cite[Theorem~2]{muller2013quantum}.
\end{proof}

The quantum $\alpha$-mutual information can be defined by expanding~\eqref{eqn:Sibson_info} to quantum states as
\begin{align*}
    I_\alpha(A;B)_{\rho_{AB}}:=\inf_{\sigma_B} D_\alpha(\rho_{AB}\|\rho_A\otimes\sigma_B),
\end{align*}
where $\rho_{AB}$ denotes the bipartite quantum state in $\mathcal{S}(\mathcal{H}_A\times\mathcal{H}_B)$ and $\rho_A=\trace_B(\rho_{AB})\in\mathcal{S}(\mathcal{H}_A)$. 
Similarly, we can define the sandwiched quantum $\alpha$-mutual information as
\begin{align*}
    \widetilde{I}_\alpha(A;B)_{\rho_{AB}}:= \inf_{\sigma_B} \widetilde{D}_\alpha(\rho_{AB}\|\rho_A\otimes\sigma_B).
\end{align*}

\section{Information Leakage to Arbitrary Eavesdropper} \label{sec:info_leakage}
Let discrete random variable $X\in\X$ model the classical data that must be kept private or secure. For any $x\in\X$, quantum system $A$ with state $\rho_A^x\in\mathcal{S}(\mathcal{H}_A)$ is prepared. The ensemble of the states $\mathcal{E}:=\{p_X(x),\rho_A^x\}_{x\in\X}$ captures the quantum encoding of the classical data. The average or expected density operator is given by $\rho_A=\mathbb{E}\{\rho_A^X\}=\sum_{x\in\X}p_X(x)\rho_A^x$. The expected density operator $\rho_A$ is the state of the system from the prospective of someone who does not know the realization of random variable $X$.

We assume that an arbitrary eavesdropper, who does not know the realization of the random variable $X$, wants to reliably guess or estimate the realization of a possibly randomized discrete function of the random variable $X$, denoted by random variable $Z$, based on measurements from a single copy of the quantum state of system $A$. The security analyst is not aware of the intention of the eavesdropper (i.e., the nature of the random variable $Z$ that is of interest to the eavesdropper). This setup also covers the case that eavesdropper searches for the random variable $Z$, i.e., an attack strategy, that results in the maximal information leakage~\cite{farokhiPRA}. For a given POVM $F=\{F_y\}_{y\in\Y}$, discrete random variable $Y\in\Y$ denotes the outcome of the measurement such that the probability of obtaining measurement outcome $Y=y\in\Y$ when taking a measurement on quantum state $\rho_A^x$ is $\mathbb{P}\{Y=y\,|\,X=x\}:=\trace(\rho_A^x F_y)$. The eavesdropper makes a one-shot\footnote{It was shown that number of guesses is immaterial in evaluating the maximal information leakage~\cite{farokhiPRA}.} guess of the random variable $Z$ denoted by the random variable $\widehat{Z}$. Maximal quantum leakage, originally defined and studied in~\cite{farokhiPRA}, measures the multiplicative increase in the probability of correctly guessing the realization of the random variable $Z$ based on access to the quantum encoding of the data  via ensemble $\mathcal{E}=\{p_X(x),\rho_A^x\}_{x\in\X}$.

\begin{definition}[Maximal Quantum Leakage] \label{def:qml} The maximal quantum leakage from random variable $X$ through quantum encoding of the data via ensemble $\mathcal{E}=\{p_X(x),\rho_A^x\}_{x\in\X}$ is 
\begin{subequations}\label{eqn:def_qml}
    \begin{align}
    \mathcal{Q}(X\!\rightarrow \!A)_{\rho_A}
    :=&\sup_{\{F_y\}_y}\sup_{Z,\widehat{Z}} \log \!\left(\!\frac{\mathbb{P}\{Z=\widehat{Z}\}}{\displaystyle \max_{z\in\mathbb{Z}}\mathbb{P}\{Z=z\}} \!\right)\label{eqn:qml:original_def}\\
    =&\sup_{\{F_y\}_y} I_\infty(X;Y) \\
    =&\sup_{\{F_y\}_y} \log\!\left(\sum_{y\in\mathbb{Y}} \max_{x\in\mathbb{X}} \trace(\rho_A^x F_y) \!\right)\!,
\end{align}
\end{subequations}
    where, in~\eqref{eqn:qml:original_def}, the inner supremum is taken over all random variables $Z$ and $\widehat{Z}$ with equal arbitrary finite support sets and the outer supremum is taken over all POVMs $F=\{F_y\}_{y\in\Y}$ with arbitrary finite set of outcomes $\Y$. 
\end{definition}

As noted in~\cite{farokhiPRA}, there is no explicit formula for the maximal quantum leakage and an iterative algorithm must be used to compute this quantity for various quantum encoding methods. This motivates developing upper bounds for maximal quantum leakage that are easier to compute. In this paper, we develop two upper bounds for the maximal quantum leakage and show that these novel quantities also satisfy important properties or axiom for measures of information leakage. They are however considerably simpler to compute. They can be either reformulated as a semi-definite program or possess an explicit form. 

\begin{proposition} \label{prop:upperboundqml} The maximal quantum leakage from random variable $X$ through quantum encoding of the data via ensemble $\mathcal{E}=\{p_X(x),\rho_A^x\}_{x\in\X}$ is upper bounded by
    \begin{align*}
        \mathcal{Q}(X\rightarrow A)_{\rho_A}\leq \min_{\pi\in\Delta(\X)} \max_{x\in\X}\widetilde{D}_\infty\left(\rho_A^x\Bigg\|\sum_{x\in\X}\pi(x)\rho_A^x\right).
    \end{align*}
\end{proposition}

\begin{proof}
For all $\alpha\geq 1/2$,  we have
\begin{align*}
	I_\alpha(X;Y)
	\leq& \frac{1}{\alpha-1}\log\left(\sum_{x\in\X} p_X(x)\trace\left( \left(\rho_A^x \rho_A^{\frac{1-\alpha}{\alpha}} \right)^\alpha \right) \right)\\
	\leq& \frac{1}{\alpha-1}\log\left(|\X|\max_{x\in\X}\trace\left( \left(\rho_A^x \rho_A^{\frac{1-\alpha}{\alpha}} \right)^\alpha \right) \right)\\
	=& \frac{\log(|\X|)}{\alpha-1}+\max_{x\in\X}\frac{1}{\alpha-1}\log\left(\trace\left( \left(\rho_A^x \rho_A^{\frac{1-\alpha}{\alpha}} \right)^\alpha \right) \right)\\
	=& \frac{\log(|\X|)}{\alpha-1}+\max_{x\in\X}\widetilde{D}_\alpha(\rho_A^x\|\rho_A)
\end{align*}
where the first inequality follows from~\cite[Eq.~(22)]{bussandri2023r}. 
Therefore,
\begin{align}
	I_\infty(X;Y)
	&=\lim_{\alpha\rightarrow \infty}I_\alpha(X;Y)\nonumber\\
	& \leq \lim_{\alpha\rightarrow \infty}\max_{x\in\X}\widetilde{D}_\alpha(\rho_A^x\|\rho_A)\nonumber\\
	& = \max_{x\in\X}\lim_{\alpha\rightarrow \infty}\widetilde{D}_\alpha(\rho_A^x\|\rho_A) \label{eqn:proof:1}\\
	& = \max_{x\in\X}\widetilde{D}_\infty(\rho_A^x\|\rho_A),\nonumber
\end{align}
where~\eqref{eqn:proof:1} follows from Lemma~\ref{lemma:swap} in Appendix~\ref{sec:tehnical_lemma} in Supplementary Material. Finally, note that $I_\infty(X;Y)$ is independent of $p_X$, therefore
\begin{align*}
	I_\infty(X;Y)
	&=\min_{p_X\in\Delta(\X)}I_\infty(X;Y)\\
	&=\min_{p_X\in\Delta(\X)} \max_{x\in\X}\widetilde{D}_\infty(\rho_A^x\|\rho_A)\\
	&=\min_{p_X\in\Delta(\X)} \max_{x\in\X}\widetilde{D}_\infty\left(\rho_A^x\Bigg\| \sum_{x'\in\X}p_X(x')\rho_A^{x'}\right).
\end{align*}
This concludes the proof. 
\end{proof}

This upper bound motivates introducing a new measure of information leakage for quantum encoding of classical data, referred to as barycentric quantum R\'{e}nyi leakage. Barycenter, a term popular in astrophysics, refers to the center of mass of two or more bodies that orbit one another. In this instance, barycenter refers to a quantum states that minimizes the worst-case distance for all quantum encoding $(\rho_A^x)_{x\in\X}$. Barycentric quantum divergences have been used in the past to measure information content of a quantum encoding, albeit when the distance cost is in a sum form~\cite{mosonyi2022geometric}.

\begin{definition}[Barycentric Quantum R\'{e}nyi Leakage] \label{def:BQRL}
    The barycentric quantum R\'{e}nyi leakage from random variable $X$ through quantum encoding of the data via ensemble $\mathcal{E}=\{p_X(x),\rho_A^x\}_{x\in\X}$ is
    \begin{align}
        \mathcal{B}(X\rightarrow A)_{\rho_A}\!:=\!\!\min_{\pi\in\Delta(\X)} \max_{x\in\X}\widetilde{D}_\infty\!\left(\rho_A^x\Bigg\|\sum_{x'\in\X}\pi(x')\rho_A^{x'}\right)\!.
    \end{align}
\end{definition}

The following corollary immediately follows from Proposition~\ref{prop:upperboundqml} and Definition~\ref{def:BQRL}. 

\begin{corollary} \label{cor:bqrl}$\mathcal{Q}(X\rightarrow A)_{\rho_A}\leq \mathcal{B}(X\rightarrow A)_{\rho_A}$.
\end{corollary}

The barycentric quantum R\'{e}nyi leakage satisfies important properties or axioms for information leakage: positivity (the information leakage is greater than or equal to zero), independence property (the information leakage is zero if the quantum state is independent of the classical data), post-processing inequality (the information leakage can only be reduced if a quantum channel is applied to the state), and unitary invariance (the information remains constant by application of a unitary operator on the quantum state). The corresponding classical properties of positivity, independence, and post-processing have been postulated to be of utmost importance in security analysis~\cite{issa2019operational}. 

\begin{theorem} \label{tho:properties_bqrl} The following properties hold for the barycentric quantum R\'{e}nyi leakage:
\begin{itemize}
    \item[(a)] \textbf{Positivity and Independence}: $\mathcal{B}(X\rightarrow A)_{\rho_A}\geq 0$ with equality if and only if $\rho_A^x=\rho_A^{x'}$ for all $x,x'\in\X$;
    \item[(b)] \textbf{Unitary Invariance}: $\mathcal{B}(X\rightarrow A)_{U\rho_AU^\dag}=\mathcal{B}(X\rightarrow A)_{\rho_A}$ for any unitary $U$;
    \item[(c)] \textbf{Data-Processing Inequality}: $\mathcal{B}(X\rightarrow A)_{\mathcal{E}(\rho_A)}\leq \mathcal{B}(X\rightarrow A)_{\rho_A}$ for any quantum channel $\mathcal{E}$.
\end{itemize}
\end{theorem}

\begin{proof}
    We start by proving part (a) by contradiction. Assume that $\mathcal{B}(X\rightarrow A)_{\rho_A}<0$. The order axiom in Lemma~\ref{lemma:properties:sandwiched} shows that $\rho_A^x<\sum_{x'\in\X}\pi(x')\rho_A^{x'}$ for all $x\in\X$ and all $\pi\in\Delta(X)$. Multiplying both sides of this inequality with $\pi(x)$ and summing over $x\in\X$, we get
    $\sum_{x\in\X}\pi(x)\rho_A^x <\sum_{x\in\X}\pi(x) \sum_{x'\in\X}\pi(x')\rho_A^{x'}=\sum_{x'\in\X}\pi(x')\rho_A^{x'}$, where the last equality follows from $\sum_{x\in\X}\pi(x)=1$. This is a contradiction and therefore, it must be that $\mathcal{B}(X\rightarrow A)_{\rho_A}\geq 0$. 

    If $\rho_A^{x}=\rho_A^{x'}$ for all $x,x'\in\X$, we get 
    \begin{align*}
        \mathcal{B}(X\rightarrow A)_{\rho_A}
        &=\min_{\pi\in\Delta(\X)} \max_{x\in\X}\widetilde{D}_\infty\left(\rho_A^x\Bigg\|\sum_{x'\in\X}\pi(x')\rho_A^{x'}\right)\\
        &=\min_{\pi\in\Delta(\X)} \max_{x\in\X}\widetilde{D}_\infty(\rho_A^x\|\rho_A^x)\\
        &=0,
    \end{align*}
    where the second equality follows from $\sum_{x\in\X}\pi(x)=1$. On the other hand, $\mathcal{B}(X\rightarrow A)_{\rho_A}=0$ implies that, for all $x\in\X$ and $\pi\in\Delta(\X)$, $\rho_A^x =\sum_{x'\in\X}\pi(x')\rho_A^{x'}$. This is only possible if $\rho_A^{x}=\rho_A^{x'}$ for all $x,x'\in\X$. 

    Now, we can prove part (b). The proof of this part follows from that 
    \begin{align*}
        \widetilde{D}_\infty&\left(U\rho_A^xU^\dag\Bigg\| \sum_{{x'}\in\X}\pi(x') U\rho_A^{x'}U^\dag\right)\\
        &=\widetilde{D}_\infty\left(U\rho_A^xU^\dag\Bigg\| U\left(\sum_{{x'}\in\X}\pi(x') \rho_A^{x'}\right)U^\dag\right)\\
        &=\widetilde{D}_\infty\left(\rho_A^x\Bigg\| \sum_{{x'}\in\X}\pi(x') \rho_A^{x'}\right),
    \end{align*}
    where the second equality follows from unitary invariance in Lemma~\ref{lemma:properties:sandwiched}. 

    Finally, we can prove part (c). Note that
    \begin{align*}
        \widetilde{D}_\infty&\left( \mathcal{E}(\rho_A^x)\Bigg\|\sum_{x'\in\X}\pi(x')\mathcal{E}(\rho_A^{x'})\right)\\
        &=
        \widetilde{D}_\infty\left( \mathcal{E}(\rho_A^x)\Bigg\|\mathcal{E}\left(\sum_{x'\in\X}\pi(x')\rho_A^{x'}\right)\right)\\
        &\leq 
        \widetilde{D}_\infty\left( \rho_A^x\Bigg\|\sum_{x'\in\X}\pi(x')\rho_A^{x'}\right),
    \end{align*}
    where the equality stems from the linearity of $\mathcal{E}$ and the inequality follows form the data-processing inequality in Lemma~\ref{lemma:properties:sandwiched}.    
\end{proof}

\begin{proposition}[Semi-Definite Programming for Barycentric Quantum R\'{e}nyi Leakage] \label{tho:sdp_bqrl} The barycentric quantum R\'{e}nyi leakage from random variable $X$ through quantum encoding of the data via ensemble $\mathcal{E}=\{p_X(x),\rho_A^x\}_{x\in\X}$ can be computed using the semi-definite programming:
\begin{align*}
    \min_{\pi\in\Delta(\X),\mu\in\mathbb{R}} & \mu, \\
    \mathrm{s.t.}\quad \; & \rho_A^x\leq \mu \sum_{x'\in\X} \pi(x') \rho_A^{x'}, \forall x\in\X.
\end{align*}
\end{proposition}

\begin{proof}
    The proof immediately follows from the definition of $\widetilde{D}_\infty$.  
\end{proof}

Theorem~\ref{tho:sdp_bqrl} shows that, although the barycentric quantum R\'{e}nyi leakage still does not posses an explicit formula, its computation is far simpler than the maximal quantum leakage (c.f.,~\cite{farokhiPRA}).

\begin{example}[Basis/Index Encoding] Consider random variable $X$ with support set $\X=\{1,\dots,2^n\}$ for some positive integer $n$. Assume that a basis or index encoding strategy is used, that is, $\rho_A^x=\ket{x}\bra{x}$ for all $x\in\X$. In~\cite{farokhiPRA}, it was shown that $\mathcal{Q}(X\rightarrow A)_{\rho_A}=n$. To compute the barycentric quantum R\'{e}nyi leakage, we can use Theorem~\ref{tho:sdp_bqrl} to demonstrate that $
    \mathcal{B}(X\rightarrow A)_{\rho_A}=\min_{\pi\in\Delta(\X)}\max_{x\in\X} ({1}/{\pi(x)})=n.$
This demonstrates that the upper bound in Corollary~\ref{cor:bqrl} can be tight in some cases. 
\end{example}

In what follows, we further simplify the upper bound in Proposition~\ref{prop:upperboundqml} to drive a simpler measure of information leakage with an explicit form. 

\begin{proposition} \label{prop:brl_prl} The barycentric quantum R\'{e}nyi leakage from random variable $X$ through quantum encoding of the data via ensemble $\mathcal{E}=\{p_X(x),\rho_A^x\}_{x\in\X}$ is upper bounded by
    \begin{align*}
        \mathcal{B}(X\rightarrow A)_{\rho_A}
        \leq \max_{x,x'\in\X} \widetilde{D}_\infty(\rho_A^x\| \rho_A^{x'}).
    \end{align*}
\end{proposition}

\begin{proof}
   Note that
   \begin{align*}
   	\widetilde{D}_\infty&\left(\rho_A^x\Bigg\| \sum_{x'\in\X}\pi(x')\rho_A^{x'}\right)\\
   	&=\log\left(\inf\left\{\mu\in\mathbb{R}:\rho_A^x\leq \mu \sum_{x'\in\X}\pi(x')\rho_A^{x'}\right\}\right)\\
   	&\leq\max_{x'\in\X}\log(\inf \{\mu\in\mathbb{R}:\rho_A^x\leq \mu \rho_A^{x'}\})\\
   	&=\max_{x'\in\X} \widetilde{D}_\infty(\rho_A^x\| \rho_A^{x'}),
   \end{align*}
   where the inequality follows from 
   \begin{align*}
   	&\left\{\mu\in\mathbb{R}:\rho_A^x\leq \mu \sum_{x'\in\X}\pi(x')\rho_A^{x'}\right\}\\
   	&\hspace{.4in}=\left\{\mu\in\mathbb{R}: \sum_{x'\in\X}\pi(x')\rho_A^x \leq \mu \sum_{x'\in\X}\pi(x')\rho_A^{x'}\right\}
   	\\
   	&\hspace{.4in}=\left\{\mu\in\mathbb{R}: \sum_{x'\in\X}\pi(x')(\rho_A^x -\mu\rho_A^{x'} )\leq 0\right\}\\
   	&\hspace{.4in}\subseteq\left\{\mu\in\mathbb{R}: \rho_A^x -\mu\rho_A^{x'}\leq 0,\forall x'\in\X\right\}\\
   	&\hspace{.4in}=\bigcap_{x'\in\X}\left\{\mu\in\mathbb{R}: \rho_A^x -\mu\rho_A^{x'}\leq 0\right\}.
   \end{align*}
   This concludes the proof.
\end{proof}

This upper bound motivates introducing another measure of information leakage for quantum encoding of classical data, referred to as pairwise quantum R\'{e}nyi leakage. This notion of information leakage possesses an explicit formula and is thus easily computable; however, as we demonstrate later, is more conservative. 

\begin{definition}[Pairwise Quantum R\'{e}nyi Leakage] \label{def:pqrl}
    The pairwise quantum R\'{e}nyi leakage from random variable $X$ through quantum encoding of the data via ensemble $\mathcal{E}=\{p_X(x),\rho_A^x\}_{x\in\X}$ is
    \begin{align}
        \mathcal{R}(X\rightarrow A)_{\rho_A}:=\max_{x,x'\in\X}\widetilde{D}_\infty(\rho_A^x\|\rho_A^{x'}).
    \end{align}
\end{definition}

The following corollary immediately follows from Corollary~\ref{cor:bqrl}, Proposition~\ref{prop:brl_prl}, and Definition~\ref{def:pqrl}.

\begin{corollary} \label{cor:pqrl}$\mathcal{Q}(X\rightarrow A)_{\rho_A}\leq \mathcal{B}(X\rightarrow A)_{\rho_A}\leq \mathcal{R}(X\rightarrow A)_{\rho_A}$.
\end{corollary}

Similarly, the pairwise quantum R\'{e}nyi leakage satisfies important properties of positivity, independence, post-processing inequality, and unitary invariance.  These properties are established in the following theorem.

\begin{theorem} \label{tho:pqrl_properties} The following properties hold for the pairwise quantum R\'{e}nyi leakage:
\begin{itemize}
    \item[(a)] \textbf{Positivity and Independence}: $\mathcal{R}(X\rightarrow A)_{\rho_A}\geq 0$ with equality if and only if $\rho_A^x=\rho_A^{x'}$ for all $x,x'\in\X$;
    \item[(b)] \textbf{Unitary Invariance}: $\mathcal{R}(X\rightarrow A)_{U\rho_AU^\dag}=\mathcal{R}(X\rightarrow A)_{\rho_A}$ for any unitary $U$;
    \item[(c)] \textbf{Data-Processing Inequality}: $\mathcal{R}(X\rightarrow A)_{\mathcal{E}(\rho_A)}\leq \mathcal{R}(X\rightarrow A)_{\rho_A}$ for any quantum channel $\mathcal{E}$.
\end{itemize}
\end{theorem}
 
\begin{proof} 
We start by proving part (a). Note that $\mathcal{R}(X\rightarrow A)_{\rho_A}= \max_{x,x'\in\X}\widetilde{D}_\infty(\rho_A^x\|\rho_A^{x'}) \geq\widetilde{D}_\infty\left(\rho_A^x\|\rho_A^x\right)=0$. If $\rho_A^{x}=\rho_A^{x'}$ for all $x,x'\in\X$, $\widetilde{D}_\infty(\rho_A^x\|\rho_A^{x'})=0$ for all $x',x\in\X$. Therefore, $\mathcal{R}(X\rightarrow A)_{\rho_A}=0$. On the other hand, if $\mathcal{R}(X\rightarrow A)_{\rho_A}=0$, it means that, for all $x,x'\in\X$, $\widetilde{D}_\infty(\rho_A^x\|\rho_A^{x'})=0$, which is only possible if $\rho_A^{x}=\rho_A^{x'}$. 

The proof for part (b) is similar to the proof of part (b) of Theorem~\ref{tho:properties_bqrl} because Lemma~\ref{lemma:properties:sandwiched} results in $\widetilde{D}_\infty(U\rho_A^xU^\dag\|U\rho_A^{x'}U^\dag) =\widetilde{D}_\infty(\rho_A^x\|\rho_A^{x'})$ for all $x,x'\in\X$. 

The proof of part (c) follows from that $\widetilde{D}_\infty( \mathcal{E}(\rho_A^x)\|\mathcal{E}(\rho_A^{x'}))
\leq 
\widetilde{D}_\infty( \rho_A^x\| \rho_A^{x'})$,
where the inequality is a consequence of Lemma~\ref{lemma:properties:sandwiched}.
\end{proof}

\setcounter{example}{0}
\begin{example}[Basis/Index Encoding (Cont'd)] Note that, because $\rho_A^x\NOTll \rho_A^{x'}$ if $x\neq x'$, we can easily see that $\mathcal{R}(X\rightarrow A)_{\rho_A}=\max_{x,x'\in\X}\widetilde{D}_\infty(\rho_A^x\|\rho_A^{x'})=\infty$. Therefore, the upper bound relating to $\mathcal{R}(X\rightarrow A)_{\rho_A}$ in Corollary~\ref{cor:pqrl} can be loose in this instance. This measure of information leakage can be conservative in general.
\end{example}

We finish this section by investigating the relationship between the barycentric and pairwise quantum R\'{e}nyi leakage with the sandwiched quantum $\alpha$-mutual information and accessible information. Consider the following classical-quantum state representing the ensemble $\mathcal{E}=\{p_X(x),\rho_A^x\}_{x\in\X}$:
\begin{align*}
    \rho_{XA}=\sum_{x\in\X} p_X(x)\ket{x}\bra{x}\otimes\rho^x_A.
\end{align*}
Recall that the sandwiched quantum $\alpha$-mutual information of the classical-quantum state $\rho_{XA}$ is given by  $
    \widetilde{I}_\alpha(X;A)_{\rho_{XA}}=\inf_{\sigma} \widetilde{D}_\alpha(\rho_{XA}\|\rho_X\otimes\sigma).$

\begin{proposition} \label{prop:sandwitched_upper} The sandwiched quantum $\alpha$-mutual information of the classical-quantum state $\rho_{XA}$ is upper bounded by
    \begin{align*}
    \widetilde{I}_\alpha(X;A)_{\rho_{XA}}\leq \widetilde{I}_\infty(X;A)_{\rho_{XA}}
    &\leq \mathcal{B}(X\rightarrow A)_{\rho_A}\\
    &\leq\mathcal{R}(X\rightarrow A)_{\rho_A}.
    \end{align*}
\end{proposition}

\begin{proof}
Firstly, $\widetilde{I}_\alpha(X;A)_{\rho_{XA}}\leq \widetilde{I}_\infty(X;A)_{\rho_{XA}}$ is a direct consequence of Theorem~7 in~\cite{muller2013quantum}. 
Note that 
\begin{align*}
	\widetilde{D}_\infty&(\rho_{XA}\|\rho_X\otimes\sigma)\\
	&=\log\left(\inf \left\{\mu\in\mathbb{R}: \rho_{XA}\leq \mu (\rho_X\otimes\sigma )\right\} \right)\\
	&\leq\log\left(\inf \left\{\mu\in\mathbb{R}: \mu\sigma-\rho^x_A\geq 0, \forall x\in\X\right\} \right)\\
	&=\max_{x\in\X}\log\left(\inf \left\{\mu\in\mathbb{R}: \mu\sigma-\rho^x_A\geq 0\right\} \right)\\
	&=\max_{x\in\X}\widetilde{D}_\infty (\rho^x_A \|\sigma),
\end{align*}
where the inequality follows from that
\begin{align*}
	&\left\{\mu\in\mathbb{R}: \rho_{XA}\leq \mu \rho_X\otimes\sigma \right\}\\
	&\hspace{.4in}=
	\Bigg\{\mu\in\mathbb{R}: \sum_{x\in\X} p_X(x)\ket{x}\bra{x}\otimes\rho^x_A\\
	&\hspace{1.3in}\leq \mu \sum_{x\in\X} p_X(x)\ket{x}\bra{x}\otimes\sigma \Bigg\}
	\\&\hspace{.4in}=
	\left\{\mu\in\mathbb{R}: \sum_{x\in\X} p_X(x)\ket{x}\bra{x}\otimes(\mu\sigma\!-\!\rho^x_A)\geq 0 \right\}
	\\&\hspace{.4in}\supseteq
	\left\{\mu\in\mathbb{R}: \mu\sigma-\rho^x_A\geq 0,\quad \forall x\in\X \right\}.
\end{align*}
This implies that 
\begin{align*}
	\widetilde{I}_\infty(X;A)_{\rho_{XA}}
	=&\inf_{\sigma} \widetilde{D}_\infty(\rho_{XA}\|\rho_X\otimes\sigma)\\
	\leq & \inf_{\sigma}\max_{x\in\X}\widetilde{D}_\infty(\rho^x_A \|\sigma)\\
	\leq&\inf_{\pi\in\Delta(\X)} \max_{x\in\X}\widetilde{D}_\infty\left(\rho^x_A \Bigg\|\sum_{x'\in\X}\pi(x')\rho^{x'}_A \right).
\end{align*}
This concludes the proof.
\end{proof}

An important notion of information in quantum information theory is accessible information~\cite[p.\,298]{wilde2013quantum}. For ensemble $\mathcal{E}=\{p_X(x),\rho_A^x\}_{x\in\X}$, the accessible information is 
\begin{align*}
    I_{acc}(\mathcal{E}):=\sup_{\{F_y\}_y} I_1(X;Y).
\end{align*}
In addition, the Holevo's information~\cite{holevo1973bounds} (also see~\cite[p.\,318]{wilde2013quantum}) is 
\begin{align*}
    \chi(\mathcal{E}):=I_1(X;A)_{\rho_A}=H(\rho_A)-\sum_{x\in\X}p_X(x)H(\rho_A^x),
\end{align*}
where $H(\rho)=-\trace(\rho\log(\rho))$ is the von Neumann quantum entropy. The next proposition provides a relationship between accessible information, and the barycentric and pairwise quantum R\'{e}nyi leakage.

\begin{proposition} \label{prop:accessible}
    $I_{acc}(\mathcal{E})\leq \chi(\mathcal{E})\leq \mathcal{B}(X\rightarrow A)_\rho\leq \mathcal{R}(X\rightarrow A)_\rho$.
\end{proposition}

\begin{proof}
First, note that the Holevo bound demonstrates that $I_{acc}(\mathcal{E})\leq I_1(X;A)_{\rho_{XA}}$~\cite{holevo1973bounds}. Furthermore, $I_1(X;A)_{\rho_{XA}}=\widetilde{I}_1(X;A)_{\rho_{XA}}$. In addition, $\widetilde{I}_1(X;A)_{\rho_{XA}}\leq \widetilde{I}_\infty(X;A)_{\rho_{XA}}$~\cite[Theorem~7]{muller2013quantum}. The rest follows from Proposition~\ref{prop:sandwitched_upper}.
\end{proof}

\section{Quantum Differential Privacy and Depolarizing Channels}
\label{sec:DP}
Differential privacy is the gold standard of data privacy analysis in the computer science literature~\cite{dwork2008differential}. Differential privacy has been recently extended to quantum computing~\cite{hirche2023quantum, zhou2017differential, aaronson2019gentle}. The aim of the differential privacy in quantum computing is to ensure than an adversary, e.g., eavesdropper, cannot distinguish between two ``similar'' datasets based on  measurements of the underlying quantum system, c.f., hypothesis-testing privacy~\cite{farokhi2023privacy}. Similarity is modelled using a neighbourhood relationship.

\begin{definition}[Neighbouring Relationship]
A neighbouring relationship (over the set of density operators) is a mathematical relation $\sim$ that is both reflective ($\rho\sim\rho$ for all density operators $\rho$) and symmetric ($\rho\sim\sigma$ implies $\sigma\sim\rho$ for any two density operators $\rho,\sigma$). 
\end{definition}

An example of neighbouring relationship can be defined using the trace distance over quantum density operators, i.e., $\rho\sim\sigma$ if $\|\rho-\sigma\|_1\leq \kappa$ for some constant $\kappa>0$~\cite{zhou2017differential}. This is formalized in the following definition.

\begin{definition}[Closeness Neighbouring Relationship]
    $\rho\sim\sigma$ if $\|\rho-\sigma\|_1\leq \kappa$.
\end{definition}

Note that other neighbouring relationships can be embedded in differential privacy. For instance, two quantum density operators can be neighbouring if they are encode two private datasets that differ in the data of one individual~\cite{dwork2008differential}. 

\begin{definition}
    For any $\epsilon,\delta\geq 0$, a quantum channel $\mathcal{E}$ is $(\epsilon,\delta)$-differentially private if 
    \begin{align}
        \trace(M\mathcal{E}(\rho))\leq \exp(\epsilon) \trace(M\mathcal{E}(\sigma))+\delta,
    \end{align}
    for all measurements $0\preceq M\preceq I$ and neighbouring density operators $\rho\sim\sigma$. 
\end{definition}

An interesting question is to measure the effect of the quantum differential privacy on information leakage.

\begin{proposition}\label{prop:upper_bound_B_P_DP}
    $\mathcal{B}(X\rightarrow A)_{\mathcal{E}(\rho_A)}\leq\mathcal{R}(X\rightarrow A)_{\mathcal{E}(\rho_A)}\leq \epsilon/\ln(2)$ if $\rho_A^x\sim\rho_A^{x'}$ for all $x,x'\in\X$ under any neighbouring relationship and quantum channel $\mathcal{E}$ is $(\epsilon,0)$-differentially private.
\end{proposition}

\begin{proof}
    From Lemma~III.2~\cite{hirche2023quantum} (with $\delta=0$), we get $\widetilde{D}_\infty(\mathcal{E}(\rho_A^x)\|\mathcal{E}(\rho_A^{x'}))\leq \epsilon/\ln(2).$ Note that the division by $\ln(2)$ is caused by the use of natural basis in definition of the differential privacy and logarithms in~\cite{hirche2023quantum}. This shows that $\mathcal{R}(X\rightarrow A)_{\rho_A}\leq \epsilon/\ln(2)$. The rest follows from Corollary~\ref{cor:pqrl}. 
\end{proof}

A physical noise model for quantum systems is the global depolarizing channel:
\begin{align} \label{eqn:dep_channel}
    \mathcal{D}_{p,d_A}(\rho):=\frac{p}{d_A}I+(1-p)\rho,
\end{align}
where $d_A$ is the dimension of the Hilbert space $\mathcal{H}_A$ to which the system belongs and $p\in[0,1]$ is a probability parameter. The larger the probability parameter $p$, the noisier the global depolarizing channel $\mathcal{D}_{p,d_A}$.

\begin{figure}
    \centering
    \begin{tikzpicture}
        \node[] at (0,0) {\includegraphics[width=1\linewidth]{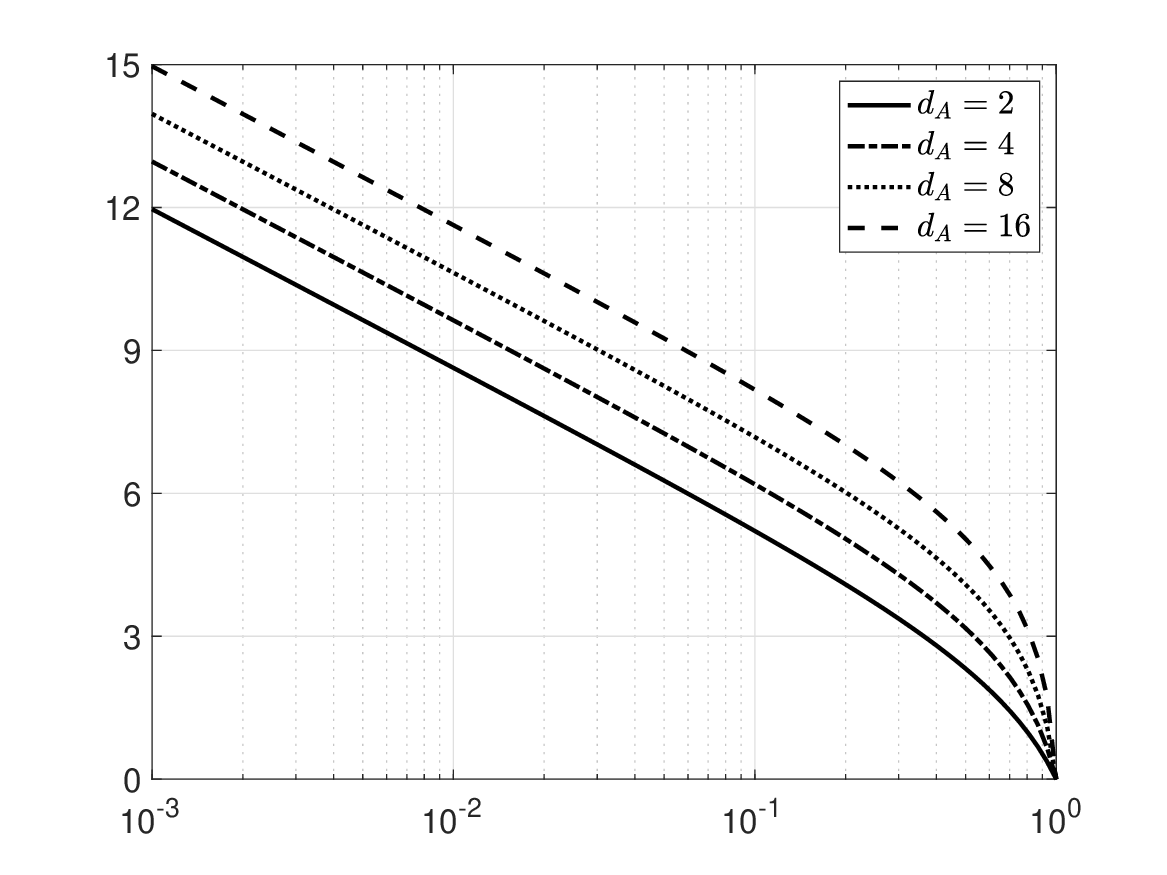}};
        \node[] at (0,-3.1) {\small probability parameter $p$};
        \node[rotate=90] at (-3.7,0) {\small $\log(1+2(1-p) d_A/p)$};
    \end{tikzpicture}
    \caption{The upper bound for $\mathcal{B}(X\rightarrow A)_{\mathcal{D}_{p,d_A}(\rho_A)}$}
    \label{fig:upper_bound}
\end{figure}

\begin{proposition} \label{prop:global}
    $\mathcal{B}(X\rightarrow A)_{\mathcal{D}_{p,d_A}(\rho_A)}\leq \mathcal{R}(X\rightarrow A)_{\mathcal{D}_{p,d_A}(\rho_A)}\leq \log(1+2(1-p) d_A/p)$.
\end{proposition}

\begin{proof}
Lemma~IV.2~\cite{hirche2023quantum} demonstrates that $\mathcal{D}_{p,d_A}$ is $(\epsilon,0)$-differentially-private with $\epsilon=\ln(1+2(1-p) d_A/p)$. Note that, here, we set $\kappa=2$ as always $\|\rho_A^x-\rho_A^{x'}\|\leq 2$ for any two density operators $\rho_A^x$ and $\rho_A^{x'}$. Finally, note that $\ln(1+2(1-p) d_A/p)/\ln(2)=\log(1+2(1-p) d_A/p)$. The rest follows from Propositions~\ref{prop:accessible} and~\ref{prop:upper_bound_B_P_DP}.
\end{proof}

Proposition~\ref{prop:global} demonstrates that by increasing $p$, the amount of leaked information, measured by both the barycentric and pairwise quantum R\'{e}nyi leakage, tends to zero. This is established for the global depolarizing channels. However, in quantum computing devices, each qubit can be affected by local noise. To model this case, assume that the Hilbert space $\mathcal{H}_A$ is composed of $k$ qubits so that $d_A=2^k$. The local depolarizing noise channel is defined as
\begin{align}
    \mathcal{D}_{p,2}^{\otimes k}:=\mathcal{D}_{p,2}\otimes\cdots\otimes \mathcal{D}_{p,2},
\end{align}
where a depolarising channel $\mathcal{D}_{p,2}$ acts locally on each qubit. 

\begin{proposition} \label{prop:local}
$\mathcal{B}(X\rightarrow A)_{\mathcal{D}_{p,2}^{\otimes k}(\rho_A)}\leq \mathcal{R}(X\rightarrow A)_{\mathcal{D}_{p,2}^{\otimes k}(\rho_A)}\leq \log(1+2(1-p) d_A/p)$.
\end{proposition}

\begin{proof}
The proof is similar to that of Proposition~\ref{prop:global} with the exception of relying on Corollary IV.3~\cite{hirche2023quantum}.
\end{proof}

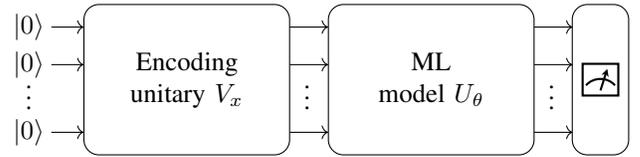
\begin{figure}
    \centering
    \begin{tikzpicture}
        \node[draw,rectangle,rounded corners=2mm,minimum width=2.5cm,minimum height=2cm]{
        \begin{minipage}{2.5cm}
        \centering 
            Encoding \\ unitary $V_x$
        \end{minipage}
        };
        \node[] at (-2.1,.7) {$\ket{0}$};
        \node[] at (-2.1,.2) {$\ket{0}$};
        \node[] at (-2.1,-.13) {$\vdots$};
        \node[] at (-2.1,-.7) {$\ket{0}$};
        \draw[->] (-1.8,.7) -- (-1.37,.7);
        \draw[->] (-1.8,.2) -- (-1.37,.2);
        \draw[->] (-1.8,-.7) -- (-1.37,-.7);
        \node[draw,rectangle,rounded corners=2mm,minimum width=2.5cm,minimum height=2cm] at (3.25,0) {
        \begin{minipage}{2.5cm}
        \centering 
            ML \\ model $U_\theta$
        \end{minipage}
        };
        \draw[->,xshift=3.25cm] (-1.88,.7) -- (-1.37,.7);
        \draw[->,xshift=3.25cm] (-1.88,.2) -- (-1.37,.2);
        \draw[->,xshift=3.25cm] (-1.88,-.7) -- (-1.37,-.7);
        \node[xshift=3.9cm] at (-2.3,-.13) {$\vdots$};
        \node[draw,rectangle,rounded corners=2mm,minimum width=.75cm,minimum height=2cm] at (5.5,0) {};
        \draw[->,xshift=6.5cm] (-1.88,.7) -- (-1.37,.7);
        \draw[->,xshift=6.5cm] (-1.88,.2) -- (-1.37,.2);
        \draw[->,xshift=6.5cm] (-1.88,-.7) -- (-1.37,-.7);
        \node[xshift=7.15cm] at (-2.3,-.13) {$\vdots$};
        \node[meter,color=black,scale=.45] at (5.5,0) {};
    \end{tikzpicture}
    \caption{A variational quantum machine learning model with encoding unitary $V_x$ and variational circuit $U_\theta$.}
    \label{fig:qml}
\end{figure}

\section{Quantum Machine Learning} \label{sec:qml}
In this section, we consider variational circuits for implementing quantum machine learning and analyzing privacy-utility trade-off in these models. The first step in a variational circuit is an encoding layer that transforms the classical data, i.e., input of the quantum machine learning model, into a quantum state. In the notation of Section~\ref{sec:info_leakage}, the ensemble $\mathcal{E}:=\{p_X(x),\rho_A^x\}_{x\in\X}$ is used to model this layer. The next layer is a variational unitary $U_\theta$ with tunable parameter $\theta$. After this layer, the state of the quantum system is $U_\theta \rho_A^xU_\theta^\dag$. Finally, measurements are taken to determine the output label. The measurement can be modelled by POVM $O=\{O_c\}$, where the outcome $c$ denotes the class to which the input belongs. Figure~\ref{fig:qml} illustrates a variational quantum machine learning model with encoding unitary $V_x$ and variational circuit $U_\theta$. Here, the quantum states are initialized at $\ket{0}\otimes\cdots\otimes\ket{0}$ and, therefore, we have $\rho_A^x=V_x\ket{0}\otimes\cdots\otimes\ket{0}\bra{0}\otimes\cdots\otimes\bra{0}V_x^\dag$. Training the quantum machine learning model entails finding parameters $\theta$, e.g., by gradient descent, to minimize the prediction error based on a training dataset. To achieve private machine learning, we can add a global or local depolarizing channel to ensure differential privacy or to bound the information leakage (in the language of this paper). The performance degradation caused by the quantum channel $\mathcal{E}$ is 
\begin{align}
    \Gamma(\mathcal{E}):=\max_{x\in\X}\sum_{c} |&\trace(O_cU_\theta \rho_A^xU_\theta^\dag)\nonumber\\
    &-\trace(O_c\mathcal{E}(U_\theta \rho_A^xU_\theta^\dag))|,
\end{align}
which captures the changes in the probability of reporting each class by addition of the quantum channel $\mathcal{E}$. The following proposition provides an upper bound for the performance degradation caused by the global depolarizing channel.

\begin{proposition} \label{prop:degradation_global}
    $\Gamma(\mathcal{D}_{p,d_A})\leq 2p$.
\end{proposition}

\begin{proof}
Note that
\begin{align*}
	|\trace(O_cU_\theta \rho_A^xU_\theta^\dag&-O_c\mathcal{D}_{p,d_A}(U_\theta \rho_A^xU_\theta^\dag))|\\
	&=|\trace(pO_cU_\theta \rho_A^xU_\theta^\dag-O_c\frac{p}{d_A}I)|\\
	&\leq p\trace(O_cU_\theta \rho_A^xU_\theta^\dag)+p\frac{1}{d_A}\trace(O_c).
\end{align*}
As a result,
\begin{align*}
	\Gamma(\mathcal{D}_{p,d_A})
	&\leq p\max_{x\in\X}\sum_{c}(\trace(O_cU_\theta \rho_A^xU_\theta^\dag)+\frac{1}{d_A}\trace(O_c))\\
	&=2p.
\end{align*}
This concludes the proof.
\end{proof}

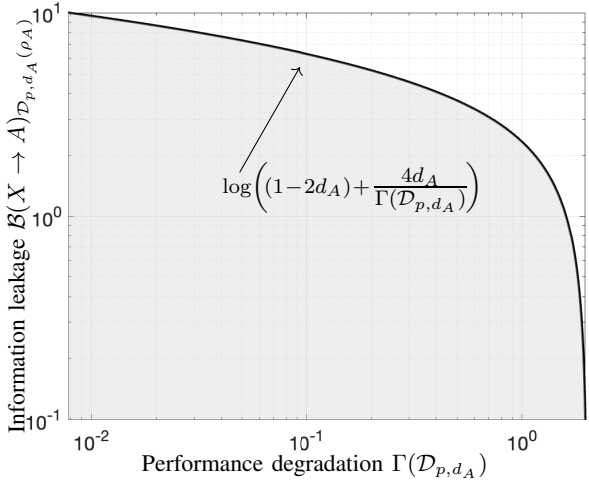
\begin{figure}
    \centering
    \begin{tikzpicture}
        \node[] at (0,0) {\includegraphics[width=1.0\columnwidth]{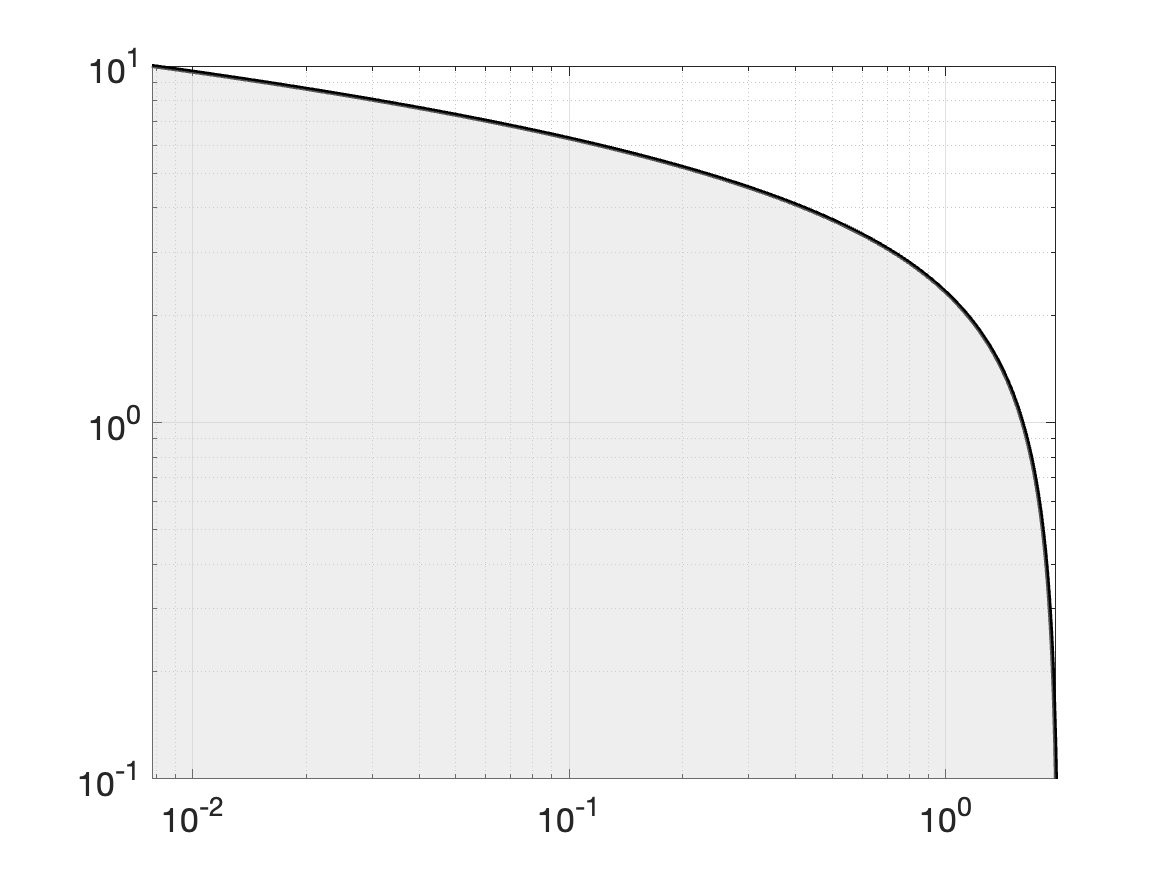}};
        \node[] at (0,-3.2) {\small Performance degradation $\Gamma(\mathcal{D}_{p,d_A})$};
        \node[rotate=90] at (-3.9,0) {\small Information leakage $\mathcal{B}(X\rightarrow A)_{\mathcal{D}_{p,d_A}(\rho_A)}$};
        \node[] at (.5,0.5) {\footnotesize $\log\!\bigg(\!(1\!-\!2d_A)\!+\!\frac{\displaystyle4d_A}{\displaystyle\Gamma(\mathcal{D}_{p,d_A})}\!\bigg)$};
        \draw[->] (-1,.7) -- (-.2,2.1);
    \end{tikzpicture}
    \caption{Privacy-utility trade-off region when using the global depolarizing channel for $d_A=2$. The solid curve depicts the upper bound in Corollary~\ref{cor:p_vs_u}.}
    \label{fig:enter-label}
\end{figure}

\begin{corollary} \label{cor:p_vs_u} The following privacy-utility trade-off holds when using the global depolarizing channel:
    \begin{align*}
        \mathcal{B}(X\rightarrow A)_{\mathcal{D}_{p,d_A}(\rho_A)}
        &\leq \mathcal{R}(X\rightarrow A)_{\mathcal{D}_{p,d_A}(\rho_A)}\\
        &\leq \log\left((1-2d_A)+\frac{4d_A}{\Gamma(\mathcal{D}_{p,d_A})}\right)
    \end{align*}
\end{corollary}

\begin{proof}
    Proposition~\ref{prop:degradation_global} shows that $1/p\leq 2/\Gamma(\mathcal{D}_{p,d_A})$. Combining this with the inequality in Proposition~\ref{prop:global} finishes the proof.
\end{proof}

Figure~\ref{fig:enter-label} illustrates the privacy-utility trade-off region when using the global depolarizing channel for $d_A=2$. The solid curve depicts the upper bound in Corollary~\ref{cor:p_vs_u}. Note that, as expected, the privacy guarantees can only be straightened, i.e., information leakage is reduced, if the performance degradation is larger. 

\section{Conclusions and Future Work}
\label{sec:conc}

Two new measures of information leakage for security and privacy analysis against arbitrary eavesdroppers, i.e., adversaries whose intention is not known by the analyst, was proposed.  
They satisfy important properties of positivity, independence, post-processing inequality, and unitary invariance. They can also be computed easily. Differentially-private quantum channels was shown to bound these new notions of information leakage. Finally, the fundamental problem of privacy-utility trade-off in quantum machine learning models was analyzed using the proposed notions of information leakage. Future work can focus on developing optimal privacy-preserving policies by minimizing the information leakage subject to a constraint on the utility. This approach is widely used in classical data privacy literature when using information-theoretic notions for security and privacy analysis. Furthermore, these measures can be used for analysis of wiretap or cipher channels, where the eavesdropper is generalized (not interested in estimating the entire secret data but can also focus on partial data recovery).

\bibliographystyle{ieeetr}
\bibliography{ref}

\begin{thebibliography}{10}

\bibitem{shor1994algorithms}
P.~W. Shor, ``Algorithms for quantum computation: Discrete logarithms and
  factoring,'' in {\em Proceedings 35th Annual Symposium on Foundations of
  Computer Science}, pp.~124--134, Ieee, 1994.

\bibitem{shor2000simple}
P.~W. Shor and J.~Preskill, ``Simple proof of security of the {BB84} quantum
  key distribution protocol,'' {\em Physical review letters}, vol.~85, no.~2,
  p.~441, 2000.

\bibitem{west2023towards}
M.~T. West, S.-L. Tsang, J.~S. Low, C.~D. Hill, C.~Leckie, L.~C. Hollenberg,
  S.~M. Erfani, and M.~Usman, ``Towards quantum enhanced adversarial robustness
  in machine learning,'' {\em Nature Machine Intelligence}, pp.~1--9, 2023.

\bibitem{biamonte2017quantum}
J.~Biamonte, P.~Wittek, N.~Pancotti, P.~Rebentrost, N.~Wiebe, and S.~Lloyd,
  ``Quantum machine learning,'' {\em Nature}, vol.~549, no.~7671, pp.~195--202,
  2017.

\bibitem{kearns2019ethical}
M.~Kearns and A.~Roth, {\em The ethical algorithm: The science of socially
  aware algorithm design}.
\newblock Oxford University Press, 2019.

\bibitem{farokhiPRA}
F.~Farokhi, ``Maximal information leakage from quantum encoding of classical
  data,'' {\em arXiv preprint arXiv:2307.12529}, 2023.

\bibitem{issa2019operational}
I.~Issa, A.~B. Wagner, and S.~Kamath, ``An operational approach to information
  leakage,'' {\em IEEE Transactions on Information Theory}, vol.~66, no.~3,
  pp.~1625--1657, 2019.

\bibitem{farokhi2021measuring}
F.~Farokhi and N.~Ding, ``Measuring information leakage in non-stochastic
  brute-force guessing,'' in {\em 2020 IEEE Information Theory Workshop (ITW)},
  pp.~1--5, 2021.

\bibitem{vrehavcek2005iterative}
J.~{\v{R}}eh{\'a}{\v{c}}ek, B.-G. Englert, and D.~Kaszlikowski, ``Iterative
  procedure for computing accessible information in quantum communication,''
  {\em Physical Review A}, vol.~71, no.~5, p.~054303, 2005.

\bibitem{muller2013quantum}
M.~M{\"u}ller-Lennert, F.~Dupuis, O.~Szehr, S.~Fehr, and M.~Tomamichel, ``On
  quantum {R}{\'e}nyi entropies: A new generalization and some properties,''
  {\em Journal of Mathematical Physics}, vol.~54, no.~12, 2013.

\bibitem{10106314906367}
K.~M.~R. Audenaert and N.~Datta, ``{$\alpha$-$z$-{R}\'{e}nyi relative
  entropies},'' {\em Journal of Mathematical Physics}, vol.~56, no.~2,
  p.~022202, 2015.

\bibitem{hirche2023quantum}
C.~Hirche, C.~Rouz{\'e}, and D.~S. Fran{\c{c}}a, ``Quantum differential
  privacy: An information theory perspective,'' {\em IEEE Transactions on
  Information Theory}, 2023.

\bibitem{renyi1961measures}
A.~R{\'e}nyi, ``On measures of entropy and information,'' in {\em Proceedings
  of the Fourth Berkeley Symposium on Mathematical Statistics and Probability,
  Volume 1: Contributions to the Theory of Statistics}, vol.~4, pp.~547--562,
  1961.

\bibitem{van2014renyi}
T.~Van~Erven and P.~Harremos, ``R{\'e}nyi divergence and kullback-leibler
  divergence,'' {\em IEEE Transactions on Information Theory}, vol.~60, no.~7,
  pp.~3797--3820, 2014.

\bibitem{cover2012elements}
J.~A. Thomas and T.~M. Cover, {\em Elements of Information Theory}.
\newblock Wiley, 2012.

\bibitem{verdu2015alpha}
S.~Verd{\'u}, ``$\alpha$-mutual information,'' in {\em 2015 Information Theory
  and Applications Workshop (ITA)}, pp.~1--6, IEEE, 2015.

\bibitem{wilde2013quantum}
M.~Wilde, {\em Quantum Information Theory}.
\newblock Cambridge University Press, 2013.

\bibitem{PETZ198657}
D.~Petz, ``Quasi-entropies for finite quantum systems,'' {\em Reports on
  Mathematical Physics}, vol.~23, no.~1, pp.~57--65, 1986.

\bibitem{umegaki1962conditional}
H.~Umegaki, ``Conditional expectation in an operator algebra, {IV} (entropy and
  information),'' in {\em Kodai Mathematical Seminar Reports}, vol.~14,
  pp.~59--85, Department of Mathematics, Tokyo Institute of Technology, 1962.

\bibitem{bhatia2013matrix}
R.~Bhatia, {\em Matrix Analysis}.
\newblock Graduate Texts in Mathematics, Springer New York, 2013.

\bibitem{mosonyi2022geometric}
M.~Mosonyi, G.~Bunth, and P.~Vrana, ``Geometric relative entropies and
  barycentric {R}\'{e}nyi divergences,'' {\em arXiv preprint arXiv:2207.14282},
  2022.

\bibitem{holevo1973bounds}
A.~S. Holevo, ``Bounds for the quantity of information transmitted by a quantum
  communication channel,'' {\em Problemy Peredachi Informatsii}, vol.~9, no.~3,
  pp.~3--11, 1973.

\bibitem{dwork2008differential}
C.~Dwork, ``Differential privacy: A survey of results,'' in {\em International
  Conference on Theory and Applications of Models of Computation}, pp.~1--19,
  Springer, 2008.

\bibitem{zhou2017differential}
L.~Zhou and M.~Ying, ``Differential privacy in quantum computation,'' in {\em
  2017 IEEE 30th Computer Security Foundations Symposium (CSF)}, pp.~249--262,
  IEEE, 2017.

\bibitem{aaronson2019gentle}
S.~Aaronson and G.~N. Rothblum, ``Gentle measurement of quantum states and
  differential privacy,'' in {\em Proceedings of the 51st Annual ACM SIGACT
  Symposium on Theory of Computing}, pp.~322--333, 2019.

\bibitem{farokhi2023privacy}
F.~Farokhi, ``Quantum privacy and hypothesis-testing,'' in {\em 2023 62nd IEEE
  Conference on Decision and Control (CDC)}, pp.~2841--2846, 2023.

\bibitem{bussandri2023r}
D.~G. Bussandri, G.~Rajchel-Mieldzio{\'c}, P.~W. Lamberti, and
  K.~{\.Z}yczkowski, ``{R}\'{e}nyi-holevo inequality from
  $\alpha$-$z$-{R}\'{e}nyi relative entropies,'' {\em arXiv preprint
  arXiv:2309.04539}, 2023.

\end{thebibliography}

\appendices
\section{Technical Lemma}
\label{sec:tehnical_lemma}
\begin{lemma} \label{lemma:swap}
	The following holds:
\begin{align*}
	\lim_{\alpha\rightarrow \infty}\max_{x\in\X}\widetilde{D}_\alpha(\rho_A^x\|\rho_A)=\max_{x\in\X}\lim_{\alpha\rightarrow \infty}\widetilde{D}_\alpha(\rho_A^x\|\rho_A).
\end{align*}
\end{lemma}

\begin{proof}
Note that $\lim_{\alpha\rightarrow \infty} \widetilde{D}_\alpha(\rho_A^x\|\rho_A)=\widetilde{D}_\infty(\rho_A^x\|\rho_A)$, which is finite because $\rho_A^x\ll \rho_A$ by definition. Therefore, for all $\alpha>0$, there exists $\bar{\alpha}_{x,\epsilon}$ such that $|\widetilde{D}_\alpha(\rho_A^x\|\rho_A)-\widetilde{D}_\infty(\rho_A^x\|\rho_A)|\leq \epsilon$ for all $\alpha\geq \bar{\alpha}_{x,\epsilon}$. This implies that, $\forall x\in\X$,
    \begin{align}\label{eqn:tildeD_inequality_proof:1}
        \widetilde{D}_\infty(\rho_A^x\|\rho_A)-\epsilon
        &\leq \widetilde{D}_\alpha(\rho_A^x\|\rho_A)\nonumber\\
        &\leq \widetilde{D}_\infty(\rho_A^x\|\rho_A)+\epsilon, \forall \alpha\geq \bar{\alpha}_{x,\epsilon}.
    \end{align}
    Based on this inequality, $\forall\alpha\geq \max_{x\in\X}\bar{\alpha}_{x,\epsilon}$, we can prove that
    \begin{align}\label{eqn:tildeD_inequality_proof:2}
        \max_{x\in\X}\widetilde{D}_\infty(\rho_A^x\|\rho_A)-\epsilon
        &\leq 
        \max_{x\in\X}\widetilde{D}_\alpha(\rho_A^x\|\rho_A)\nonumber\\
        &\leq \max_{x\in\X}\widetilde{D}_\infty(\rho_A^x\|\rho_A)+\epsilon.
    \end{align}
    The proof for this come from a contrapositive argument. For proving the upper bound, assume that there exists $\alpha\geq \max_{x\in\X}\bar{\alpha}_{x,\epsilon}$ such that $\max_{x\in\X}\widetilde{D}_\alpha(\rho_A^x\|\rho_A)>\max_{x\in\X}\widetilde{D}_\infty(\rho_A^x\|\rho_A)+\epsilon$. Therefore, there exists $x'\in\X$ such that $\widetilde{D}_\alpha(\rho_A^{x'}\|\rho_A)>\widetilde{D}_\infty(\rho_A^x\|\rho_A)+\epsilon$ for all $x\in\X$. This implies that $\widetilde{D}_\alpha(\rho_A^{x'}\|\rho_A)>\widetilde{D}_\infty(\rho_A^{x'}\|\rho_A)+\epsilon$, which is in contradiction with~\eqref{eqn:tildeD_inequality_proof:1}. For proving the lower bound, assume that there exists $\alpha\geq \max_{x\in\X}\bar{\alpha}_{x,\epsilon}$ such that $\max_{x\in\X}\widetilde{D}_\infty(\rho_A^x\|\rho_A)-\epsilon
        >
        \max_{x\in\X}\widetilde{D}_\alpha(\rho_A^x\|\rho_A)$. Therefore, there exists $x'\in\X$ such that $\widetilde{D}_\infty(\rho_A^{x'}\|\rho_A)-\epsilon
        >
        \widetilde{D}_\alpha(\rho_A^x\|\rho_A)$ for all $x\in\X$. This implies that $\widetilde{D}_\infty(\rho_A^{x'}\|\rho_A)-\epsilon
        >
        \widetilde{D}_\alpha(\rho_A^{x'}\|\rho_A)$, which is in contradiction with~\eqref{eqn:tildeD_inequality_proof:1}. 
    Equation~\eqref{eqn:tildeD_inequality_proof:2} shows that, for all $\epsilon>0$, there exists $\bar{\alpha}_{\epsilon}:=\max_{x\in\X}\bar{\alpha}_{x,\epsilon}$ such that 
    \begin{align*}
        |\max_{x\in\X}\widetilde{D}_\infty(\rho_A^x\|\rho_A)- \max_{x\in\X}\widetilde{D}_\alpha(\rho_A^x\|\rho_A)|
        \leq \epsilon,\forall\alpha\geq \bar{\alpha}_{\epsilon}.
    \end{align*}
    In addition, $\bar{\alpha}_{\epsilon}=\max_{x\in\X}\bar{\alpha}_{x,\epsilon}<\infty$ because $\X$ is finite. This concludes the proof.
\end{proof}
 
\end{document}